\documentclass[letterpaper,11pt]{article}
\usepackage[margin=1in]{geometry}
\usepackage[utf8]{inputenc}
\usepackage[english]{babel}
\usepackage{amsmath}
\usepackage{amssymb}
\usepackage{amsthm}
\usepackage{authblk}
\usepackage{tcolorbox}

\newtheorem{thm}{Theorem}

\newtheorem{lem}{Lemma}
\newtheorem{corollary}{Corollary}
\newtheorem{cor}{Corollary}

\def\ve{\varepsilon}
\def\s{\sigma}
\def\mcM{\mathcal M}
\def\mcF{\mathcal M}

\title{Stochastic Games with Limited Public Memory}
\author{}
\date{}

\author[1]{Kristoffer Arnsfelt Hansen}
\author[2]{Rasmus Ibsen-Jensen}
\author[3]{Abraham Neyman}
\affil[1]{Aarhus University, arnsfelt@cs.au.dk}
\affil[2]{University of Liverpool, R.Ibsen-Jensen@liverpool.ac.uk}
\affil[3]{Hebrew University, aneyman@huji.ac.il}

\begin{document}
\maketitle
\thispagestyle{empty}

\begin{abstract}
We study the memory resources required for near-optimal play in two-player zero-sum stochastic games with the long-run average payoff. Although optimal strategies may not exist in such games, near-optimal strategies always do.

A \emph{memory-based strategy} selects an action at each stage based on the current game state, stage number, and memory state. The memory state, which summarizes the past play, is updated stochastically at each stage as a function of the current play, the current memory state, and the stage number.
A \emph{public-memory strategy} is a memory-based strategy in which the opponent is allowed to condition her actions on the player's current memory state.

Mertens and Neyman (1981) proved that in any stochastic game, for any $\varepsilon>0$, there exist uniform $\varepsilon$-optimal memory-based strategies---i.e., strategies that are $\varepsilon$-optimal in all sufficiently long $n$-stage games---that use at most $O(n)$ memory states within the first $n$ stages. We improve this bound on the number of memory states by proving that in any stochastic game, for any $\varepsilon>0$, there exist uniform $\varepsilon$-optimal memory-based strategies that use at most $O(\log n)$ memory states in the first $n$ stages. Moreover, we establish the existence of uniform $\varepsilon$-optimal memory-based strategies whose memory updating and action selection are time-independent and such that, with probability close to 1, for all $n$, the number of memory states used up to stage $n$ is at most $O(\log n)$.

This result cannot be extended to strategies with bounded public memory---even if time-dependent memory updating and action selection are allowed. This impossibility is illustrated in the Big Match---a well-known stochastic game where the stage payoffs to Player 1 are 0 or 1. Although for any $\varepsilon > 0$, there exist strategies of Player 1 that guarantee a payoff {exceeding} $1/2 - \varepsilon$ in all sufficiently long $n$-stage games, we show that any strategy of Player 1 that uses a finite public memory fails to guarantee a payoff greater than $\varepsilon$ in any sufficiently long $n$-stage game.

\end{abstract} 
\newpage\setcounter{page}{1}

\section{Introduction}
One of the fundamental questions in computer science concerns the computational resources required to solve complex problems, with particular focus on time and space complexity. In decision-making settings, memory plays a crucial role in determining whether near-optimal strategies can be computed and implemented efficiently. This is especially the case in general models of competitive multi-stage interactions such as stochastic games. Stochastic games finds applications in diverse scientific areas, including computer science. A few examples include synthesis of synchronized programs~\cite{dAHM00,dAHM01}, preventing attacks on crypto-currency protocols~\cite{CGIV18}, radio networks~\cite{FS08}, submarine warfare~\cite{CS67}, and explaining how cooperation can arise in nature~\cite{HSCN18}. Further applications in economics and related fields are surveyed by Amir~\cite{Amir02}.

In this paper we address the fundamantal challenge of quantifying the memory resources required for near-optimal play in stochastic games.

\subsection*{Stochastic Games and the Trade-off between Short- and Long-Term Payoffs}

A \emph{stochastic game}, introduced by Shapley~\cite{sha53}, is a two-player zero-sum multistage game where the state evolves over time based on the players' actions. The game proceeds in discrete stages $t = 1,2,\dots$, with each stage beginning in one of finitely many states. In every stage, each player selects an action from a finite set, and the resulting stage payoff $r_t$ (to player 1, the maximizer) and the transition probabilities to the next state depend on the current state and the chosen actions. Crucially, while each player observes the current state and all past actions, he must select his action simultaneously with the other player, without knowing the other player's choice of action for that stage.

Shapley’s framework has been extended in several directions, including to games with infinitely many states and actions and to multi-player, non-zero-sum settings. This paper focuses exclusively on \textbf{two-player zero-sum stochastic games with finitely many states and actions}, referred to henceforth simply as \emph{stochastic games}.

A defining feature of stochastic games is the trade-off between short-term and long-term objectives: in any given stage, a player must balance between maximizing his short-term payoffs and  influencing the game’s future states in order to maximize his long-term payoffs. This tension is fundamental in all models of stochastic games but takes on different characteristics depending on how payoffs are aggregated over time.

\subsection*{Three Payoff Models in Stochastic Games}

Three primary models are used to evaluate payoffs in stochastic games, each leading to distinct strategic considerations:

\begin{enumerate}
    \item \textbf{Discounted Payoff Model}: The \emph{$\lambda$-discounted game} assigns a payoff
    \[
    \sum_{t=1}^{\infty} \lambda (1 - \lambda)^{t-1} r_t,
    \]
    where $0 < \lambda \leq 1$ is the discount rate. Here, earlier payoffs are weighted more heavily, and the strategic trade-off between short- and long-term payoffs is \emph{independent of the stage number}. Shapley~\cite{sha53} proved that every discounted stochastic game has a well-defined value and that each player has an optimal \emph{stationary strategy}, i.e., a strategy whose choice of action depends only on the current state. Bewley and Kohlberg~\cite{bewkoh76} proved that the value of the $\lambda$-discounted game converges as $\lambda$ goes to 0.

    \item \textbf{Finite-Horizon Model}: The \emph{$n$-stage game} evaluates payoffs using the average
    \[
    \overline{r}_n = \frac{r_1 + \dots + r_n}{n}.
    \]
    Unlike in the discounted model, here the balance between short-term and long-term payoffs \emph{depends on the number of remaining stages}. It follows from backward induction that the current actions in optimal strategies for the $n$-stage game {can} depend only on the current state and the number of remaining stages. Bewley and Kohlberg~\cite{bewkoh76} proved that the value of the $n$-stage game converges as $n$ goes to $\infty$ and that the limit equals the limit of the values of the $\lambda$-discounted games as $\lambda$ goes to $0$. This common limit is the value of the stochastic game.

    \item \textbf{Long-Run Average Payoffs in Stochastic Games}: There are two main approaches to study stochastic games with long-run average payoffs. The {\em uniform approach} and the  {\em undiscounted  approach}.  {In the uniform approach, the game is viewed as either a finite-horizon game with an uncertain large number of stages, or a discounted game with an uncertain small discount rate. Each player aims to perform well in every sufficiently long finite-horizon game and for every sufficiently small discount rate. Throughout this paper, we refer to Player 1 as the maximizing player, and Player 2---interchangeably called the opponent or the adversary---as the minimizing player. In the undiscounted approach, the objective is to optimize a specific long-run average payoff criterion, such as the limit inferior or limit superior of the average payoff over the first $n$ stages.}

        A near-optimal strategy in the uniform approach is called a \emph{uniform $\varepsilon$-optimal strategy}, which is a strategy that is $\varepsilon$-optimal in all sufficiently long finite-horizon games.

        In the \emph{undiscounted case}, two extreme forms of long-run average payoffs are the \emph{limit superior} and \emph{limit inferior} of $\overline{r}_n$ as $n$ goes to infinity. Each one of these long-run average payoffs leads to a corresponding concept of a near-optimal strategy:\\
    - A \emph{$\limsup$ $\varepsilon$-optimal strategy} for Player 1 guarantees that the expectation of $\limsup_{n\to\infty}\overline{r}_n$ is at least the value of the stochastic game minus $\varepsilon$.\\
    - A \emph{$\liminf$ $\varepsilon$-optimal strategy} for Player 1 guarantees that the expectation of $\liminf_{n\to\infty}\overline{r}_n$ is at least the value of the stochastic game minus $\varepsilon$.

    The key challenge in long-run average payoff models is that, while players must still balance short- and long-term payoffs, there is no natural ``horizon'' to structure strategy adjustments, making the design of near-optimal strategies significantly more difficult. Mertens and Neyman~\cite{merney81,merney82} proved that in every stochastic game, Player 1 has, for every $\varepsilon>0$, a strategy that is \emph{both} uniform $\varepsilon$-optimal and $\liminf$ $\varepsilon$-optimal, ensuring near-optimality for all long-run average payoffs.
\end{enumerate}

\subsection*{Computational Complexity of Stochastic Games}

{While this paper focuses on the memory resources required for near-optimal play, another central area of research concerns the algorithms and computational complexity of determining the value of discounted~\cite{EY08,BFGMS24} and
limit-average~\cite{CMH08,HKLMT11,KIJM11,M21,BIJT24} stochastic games.
Computing the exact value is known to lie in PSPACE~\cite{EY08,BIJT24} in both cases. The best-known approximation algorithm is in $\mathrm{FNP}^{\mathrm{NP}}$~\cite{BIJT24} for the limit-average case, and in UEOPL~\cite{BFGMS24} for the discounted case. These results underscore the broader algorithmic challenges involved in stochastic games, which are complementary to the resource-focused questions addressed in this paper.
}

\subsection*{The Big Match and the Complexity of Stochastic Games with Long-run Average Payoffs}

Gillette~\cite{gil57} introduced \emph{undiscounted stochastic games}, where the payoff is a long-run average of stage payoffs. A well-known example, \emph{the Big Match}, illustrates the fundamental difficulty of balancing between maximizing short- and long-term payoffs.

In the Big Match:
\begin{itemize}
    \item Player 2 chooses 0 or 1 in each stage, and Player 1 attempts to predict Player 2's choice.
    \item Player 1 earns a point for each correct prediction.
    \item However, if Player 1 ever predicts 1, the game transitions to an absorbing state where all future payoffs are either 0 or 1, depending on whether Player 1’s prediction was correct at that stage.
\end{itemize}

For both the \textbf{finite-horizon} and \textbf{discounted} versions of the Big Match, the value of the game is $1/2$, and optimal strategies are well understood and can be explicitly computed. However, despite its simple structure, the Big Match with long-run average payoffs exhibits \emph{severe strategic complexities}. Unlike in the discounted game, where \textbf{stationary optimal strategies exist}, and the finite-horizon game, where \textbf{Markov optimal strategies exist}, any long-run average payoff setting \textbf{requires near-optimal strategies to incorporate the memory of past play}.

Moreover, \textbf{any strategy of Player 1 that guarantees a long-run average payoff that is larger than zero must base its choice of actions on past play}.

The Big Match is a special case of an \emph{absorbing game}—a stochastic game with a single nonabsorbing state—where play either eventually reaches an absorbing state or may continue indefinitely in the nonabsorbing state.

\subsection*{The Role of Memory in Near-Optimal Strategies}
Blackwell and Ferguson~\cite{blafer68} established strategies in the Big Match that are near-optimal for all long-run average payoffs. Kohlberg~\cite{koh74} extended this result to all absorbing games and Mertens and Neyman~\cite{merney81,merney82} extended this result to all stochastic games.

%%%%%%%%%%%%%%%%%%%%%%%%%%%%%%%%%%%%%%%%%%%%
A \emph{memory-based strategy} is a strategy in which the choice of action depends on the current game state, the current stage number, and the current memory state.  The memory state, which serves as a summary of past play, is updated stochastically at each stage as a function of the stage number, the current memory state, and the current play.

{A \emph{stationary strategy} is a strategy in which action selection is time-independent, while a \emph{Markov strategy} allows the choice of actions to depend on the current stage (i.e., time-dependent). Both can be viewed as special cases of memory-based strategies that use a single memory state---effectively, strategies without memory.}

For several classes of stochastic games, such as stochastic games with perfect information and irreducible stochastic games, there exist stationary strategies that are near-optimal in the game with the long-run average payoff. However, in other stochastic games, such as the Big Match, no Markov strategy is near-optimal in any long-run average payoff setting, demonstrating the necessity of more complex memory structures.

Prior work~\cite{blafer68,koh74,merney81,merney82} presented near-optimal strategies in stochastic games with long-run average payoffs. In these near-optimal strategies, the number of memory states used up to stage $n$ grows linearly in $n$.

%%%%%%%%%%%%%%%%%%%%%%%%%%%%%
Recent work has begun to quantify and reduce the memory requirements for near-optimal play in the \emph{Big Match} and absorbing games with long-run average payoffs: Hansen, Ibsen-Jensen, and Kouchy~\cite{hanibskou16} showed that there exist near-optimal strategies that use only  {$(\log n)^{O(1)}$} memory states up to stage $n$ (this result was stated in~\cite{hanibskou16} as using $O(\log\log n)$ \emph{bits} of memory). Subsequently, Hansen, Ibsen-Jensen, and Neyman~\cite{hanibsney18, hanibsney23} demonstrated the existence of near-optimal strategies that use only finitely many memory states.

This paper quantifies and reduces memory requirements for near-optimal play in \emph{any} stochastic game with long-run average payoffs.
The main result implies that there exist, for every $\varepsilon>0$,  uniform $\varepsilon$-optimal memory-based strategies that use only $O(\log n)$ memory states up to stage $n$.

In this paper, we measure memory usage by counting the number of distinct memory states available to a strategy, rather than the number of bits required to encode these states. The two are, of course, related logarithmically: a strategy using $M$ memory states can be implemented using only $\log_2 M$ bits. Thus, our  bounds on the sufficient number of memory states are strictly stronger than analogous bounds stated in terms of memory space (bit complexity). In particular, a bound of $O(\log n)$ and a bound of $(O(\log n))^{O(1)}$ memory states both implies a bound of $O(\log \log n)$ bits of memory.

The strategies we construct are fundamentally different from the
previous limited memory strategies~\cite{hanibskou16,hanibsney18,
  hanibsney23}. In addition to only applying to the specific case of
the Big Match (or more generally absorbing games), these previous
strategies all make critical use of keeping the memory state
\emph{private}. The strategies of~\cite{hanibsney18, hanibsney23}
additionally rely on having access to the current stage number.

\subsection*{Time-dependent and time-independent choice of action and memory updating.} Different properties of memory-based strategies influence their simplicity and implementation.  The choice of action and the memory updating in any stage, which depends on the state of the game and memory state, can each be time-dependent, i.e., depending on the stage number,  or time-independent.

\subsection*{Public memory.}
The memory updating of a memory-based strategy can be deterministic or probabilistic.
A memory-based strategy whose memory updating is deterministic, enables the opponent to deduce from the memory-based strategy and the observed play the current memory state. Therefore, the memory states are necessarily public, which allows the opponent to condition his current action on the current memory state.

A memory-based strategy whose memory updating is probabilistic, enables the player to conceal the memory state from the opponent. Making the memory state public -- e.g., by using public random numbers for the probabilistic memory updating, or not concealing the updated memory state --  eliminates the need to conceal them from the opponent, simplifying implementation but potentially increasing strategic vulnerability---an important consideration in cybersecurity and adversarial decision-making. Section \ref{sec:pubvspri} discusses additional on public versus private memory

The earlier contributions~\cite{blafer68,koh74,merney81,merney82} presented near-optimal strategies in stochastic games with long-run average payoffs that use infinitely many memory states, with deterministic and time-independent memory updating, and with time-independent choice of action.

The near-optimal strategies of \cite{hanibskou16}, are time-independent and those of \cite{hanibsney18, hanibsney23} are time-dependent. The memory updating of the near-optimal strategies of \cite{hanibskou16,hanibsney18, hanibsney23}  is  probabilistic and the memory states are private.

\subsection*{Contributions of This Paper}
This paper advances the understanding of the public memory resources needed for near-optimal strategies in stochastic games with long-run average payoffs by proving the following main result:

\begin{itemize}
    \item In any stochastic game, each player has, for every $\varepsilon>0$,  uniform  $\varepsilon$-optimal public-memory strategies that, with probability close to 1, for all $n$, use at most $O(\log n)$ public memory states in the first $n$ stages.
    \item Moreover, both the memory updating and the choice of actions in these strategies are \emph{time-independent}.
\end{itemize}

This result further implies that in any stochastic game, each player has, for every $\varepsilon>0$, uniform $\varepsilon$-optimal public-memory strategies with time-dependent memory updating and time-independent choice of actions that use at most $O(\log n)$ public memory states in the first $n$ stages.

In contrast, we establish a strong worthlessness property of public-finite-memory strategies in the Big Match:
\begin{itemize}
    \item Any strategy of Player 1 in the Big Match that uses only finitely many public memory states cannot guarantee a long-run average payoff greater than 0 in any sufficiently long finite-horizon games.
    \item Moreover, for any finite-public-memory strategy of Player 1 in the Big Match and any $\varepsilon > 0$, there exists a strategy of Player 2 that yields a long-run average payoff of less than $\varepsilon$ in all sufficiently long finite-horizon games.
\end{itemize}

Even a weaker version of this worthlessness property, together with \cite[Section 3.2.1]{rasmusthesis}, implies that in the Big Match, any memory-based strategy of Player 1 with \emph{deterministic and time-independent memory updating} that guarantees a strictly positive payoff in infinitely many finite-horizon games must use at least $\Omega(n)$ memory states in the first $n$ stages.

Thus, our results highlight the fundamental advantage of \emph{probabilistic} memory updating over deterministic memory updating: while probabilistic time-independent memory updating enables near-optimal strategies that, with high probability, use at most $O(\log n)$ memory states, deterministic time-independent memory updating in the Big Match requires at least $\Omega(n)$ memory states to achieve similar guarantees in infinitely many finite-horizon games.

{While our results are stated for two-player zero-sum stochastic games, they naturally extend to the analysis of the minmax and maxmin values of a player in multi-player non-zero-sum stochastic games. This extension follows the classic approach of viewing the player of interest as the maximizing Player 1 and treating the group of all other players as the minimizing Player 2 in a two-player zero-sum reformulation.}

\section{The Stochastic Game Model and Memory-Based Strategies}
\subsection{Stochastic games}
A {\em two-person zero-sum  stochastic game} $\Gamma$, henceforth, a {\em stochastic game}, is defined by a tuple $(Z,I,J,r,p)$, where $Z$ is a finite state space, $I$ and $J$ are the finite actions sets of Players 1 and 2 respectively, $r:Z\times I\times J\to \mathbb{R}$ is a payoff function, and $p: Z\times I\times J\to \Delta(Z)$ is a transition function.

A state $z\in Z$ is called an {\em absorbing state} if $p(z,\cdot,\cdot)=\delta_{z}$, where $\delta_{z}$ is the Dirac measure on $z$.
An {\em absorbing game} is a stochastic games in which there exists exactly one state that is not absorbing.

 A {\em play} of the stochastic game is an infinite sequence $z_1,\ldots,z_t,i_t,j_t, \ldots$, where $(z_t,i_t,j_t)\in Z\times I\times J$. The set of all plays is denoted by $H_\infty$. A play up to stage $t$ is the finite sequence $h_t=(z_1,i_1,j_1,\ldots,z_t)$. The payoff $r_t$ in stage $t$ is $r(z_t,i_t,j_t)$ and the average of the payoffs in the first $n$ stages, $\frac1n\sum_{t=1}^{n}r_t$, is denoted by $\bar{r}_n$.

The initial state of the multi-stage game is $z_1\in Z$. In the $t$-th stage players simultaneously choose actions $i_t\in I$ and $j_t\in J$.

A behavioral strategy of Player 1, respectively Player 2, is a function $\s$, respectively $\tau$, from the disjoint union $\mathbin{\dot{\cup}}_{t=1}^{\infty} (Z\times I\times J)^{t-1}\times Z$ to $\Delta(I)$, respectively to $\Delta(J)$. The restriction of $\s$, respectively $\tau$, to $(Z\times I\times J)^{t-1}\times Z$ is denoted by $\s_t$, respectively $\tau_t$.
In what follows, $\s$ denotes a strategy of Player 1 and $\tau$ denotes a strategy of Player 2.

A strategy pair $(\s,\tau)$ defines a probability distribution $P_{\s,\tau}$ on the space of plays as follows. The conditional probability of $(i_t=i,j_t=j)$ given the play $h_t$ up to stage $t$ is the product of $\s(h_t)[i]$ and $ \tau(h_t)[j]$. The conditional distribution of $z_{t+1}$ given $h_t,i_t,j_t$ is $p(z_t,i_t,j_t)$. Given a strategy pair $(\s,\tau)$, the induced probability over plays is $P_{\s,\tau}$, and the expectation of a random variable $X$ under this distribution is denoted by $E_{\s,\tau}X$.

A stochastic game has a {\em value} $v=(v(z))_{z\in Z}$ if, for every $\ve>0$, there are strategies $\s_{\ve}$ and $\tau_{\ve}$ such that for some positive integer $n_\ve$
\begin{equation}\label{sgeoptimaluniform}\ve+E_{\s_\ve,\tau}\bar{r}_n\geq v(z_1)\geq E_{\s,\tau_\ve}\bar{r}_n-\ve \;\;\;\forall \s, \tau, n\geq n_\ve,
\end{equation}
and
\begin{equation}\label{sgeoptimallaverage}\ve+E_{\s_\ve,\tau}\liminf_{n\to \infty}\bar{r}_n\geq v(z_1)\geq E_{\s,\tau_\ve}\limsup_{n\to \infty}\overline{r}_n-\ve \;\;\;\forall \s, \tau.\end{equation}

It is known that all absorbing games \cite{koh74,merneyros09} and, more generally,  all finite stochastic games \cite{merney81,merney82} have a value.

A strategy $\s_\ve$ that satisfies the left-hand inequality (\ref{sgeoptimaluniform})  is called {\em uniform  $\ve$-optimal}.
A strategy $\s_\ve$ that satisfies the left-hand inequality  (\ref{sgeoptimallaverage})  is called   {\em limiting-average $\ve$-optimal}.

A strategy $\s_\ve$ that satisfies both left-hand inequalities  (\ref{sgeoptimaluniform}) and (\ref{sgeoptimallaverage})  is called {\em $\ve$-optimal}.

\subsection{Memory-based strategies} A {\em memory-based strategy} $\s$ generates a random sequence of memory states $m_1,\ldots,m_t,m_{t+1},\ldots$, where at each stage $t$, the memory state $m_t$ is updated stochastically according to a distribution that depends only on the current stage $t$ and the current game state $z_t$, as well as on the previous memory state $m_{t-1}$ and the action pair $(i_{t-1},j_{t-1})$. At each stage $t$, the action $i_t$ is chosen according to a distribution that depends only on the current time $t$, the current memory state $m_t$, and the current game state $z_t$. Explicitly, the conditional distribution of  $i_t$, given $h^m_t:=(z_1,m_1,i_1,j_1,\ldots, z_t,m_t)$, is a function $\s_\alpha$ of $(t,z_t,m_t)$  and the conditional distribution of $m_{t+1}$, given $(h^m_t,i_t,j_t,z_{t+1})$, is a function $\s_m$ of $(t,m_t,i_t,j_t,z_{t+1})$ (i.e., it depends on just the time $t$ and the tuple $(m_t,i_t,j_t,z_{t+1})$).

A memory-based strategy $\s$ is  {\em clock-independent} if its action-selection function $\s_\alpha$ and memory-update function $\s_m$ do not depend on the stage number $t$.

A natural question is the existence of memory-based strategies where the number of distinct memory states used in the first $n$ states grows slowly with high probability.

A {\em public-memory strategy} is a memory-based strategy in which, after each update, the new memory state $m_t$ is publicly revealed, allowing the opponent to base their choice of action at stage $t$ on $m_t$. Such a strategy of the other player is an $(\overline{m}_t)$-based strategy, where $\overline{m}_t=(z_1,m_1,i_1,j_1,\ldots, z_t,m_t)$.

A {\em memory-process} for a stochastic game is an $\mathbb{N}$-valued stochastic process $(m_t)_{t=1}^\infty$ where each memory state $m_t$ is updated stochastically based on past play.
The initial memory state $m_1$ is (w.l.o.g.) $0$ and the conditional distribution of $m_{t+1}$ given $(h^m_t,i_t,j_t,z_{t+1})$ is  a function $\s_m$ of $(t,m_t,i_t,j_t,z_{t+1})$ (i.e., it depends on just the time $t$ and the tuple $(m_t,i_t,j_t,z_{t+1})$).
One could instead allow the conditional distribution of $m_{t+1}$ to also depend on $z_t$. However, since the number of game states is finite, $z_t$  can be encoded into the memory state $m_t$ without affecting the results of this paper.

A {\em stationary memory-process} is a memory process  $(m_t)_{t=1}^\infty$ where the memory updating function is independent of $t$.

An  $(m_t)_{t=1}^\infty$-based strategy chooses the action at stage $t$ as a function of $(t,z_t,m_t)$.
Given an $(m_t)_{t=1}^\infty$-based strategy $\sigma$ of Player 1 and an $(\overline{m}_t)_{t=1}^\infty$-based strategy $\tau$ of Player 2,  the induced probability distribution over plays and memory sequences is denoted by $P_{\s,\tau}$, and the expectation with respect to $P_{\s,\tau}$ is denoted by $E_{\s,\tau}$.

\subsection{Public vs. Private Memory in Stochastic Games}\label{sec:pubvspri}

A memory-based strategy is one where the choice of actions depends not only on the current game state but also on a memory state that is updated as the game progresses. A key distinction arises between \textbf{public-memory strategies} and \textbf{private-memory strategies}.

\paragraph{\bf Definition.}
A strategy of a player uses \emph{public memory} if the opponent have access to its memory state---either through explicit revelation or implicit derivation. This access enables the opponent to condition their choice of action on the player's memory state. Conversely, a strategy uses \emph{private memory} if the opponent does not have access to its memory state.

Public-memory is an unusual property for strategies. Unlike properties such as finite-memory, Markov, or deterministic strategies—each of which imposes constraints or benefits on the player following the strategy—public-memory does not directly affect the player's own strategic options. Instead, it expands the \textbf{opponent’s} ability to respond by allowing them to condition their actions on the player’s memory state. This difference has significant implications in adversarial settings.

\paragraph{\bf Implications of Public vs. Private Memory.}
The distinction between public and private memory strategies affects security, robustness, and practical implementations in applications where strategies must be executed over long time horizons.

In real-world applications, strategies may be implemented in computer systems that are designed to run indefinitely in uncertain environments. A common best practice in computer science is to use backup systems distributed across multiple locations to ensure resilience against failures. However, whether these backups enhance or weaken security depends on whether the underlying strategy relies on public or private memory:
\begin{itemize}
    \item If a strategy is a \emph{private-memory strategy}, then having backups can \textbf{weaken security} in adversarial settings. An opponent only needs access to \emph{any} copy of the memory state (from the main system or a backup) to exploit weaknesses in the strategy.
    \item If a strategy is a \emph{public-memory strategy}, then backups can \textbf{enhance security}. If the system synchronizes to the most common memory state across all backups, an adversary must compromise  \emph{multiple} systems to disrupt the strategy.
\end{itemize}
Additionally, if an adversary only has \emph{read access} to the system's memory state, then public-memory strategies provide no additional risk, as long as they remain well-designed.

\paragraph{\bf Connection to Public Randomness.}
The concept of public memory is related to \textbf{public randomness} in communication complexity. One can think of public memory as using public randomness to select memory states. However, there are notable differences:
\begin{enumerate}
    \item In a public-memory strategy, public randomness is used only for memory updating, whereas in a {\em fully public-random strategy}, all randomness in the system is governed by public signals. In stochastic games, if actions were selected solely based on public randomness, strategies would effectively become deterministic—insufficient even in simple cases such as rock-paper-scissors.
    \item In communication complexity, public randomness is typically a cooperative tool that enhances efficiency, allowing players to coordinate their strategies. In stochastic games, however, public memory provides the opponent with additional information, expanding their strategic options and making near-optimal play more challenging. On the other hand, public memory can be easier to implement, as it does not require securing hidden internal memory states.
\end{enumerate}

The distinction between public and private memory is fundamental to the design of near-optimal strategies in stochastic games. While in the Big Match there exist finite-private-memory strategies that are near-optimal \cite{hanibsney23}, Theorem \ref{thm:limsupEworthless} shows that any finite-public-memory strategy of Player 1 in the Big Match is worthless. While \cite{hanibsney23} establishes the existence of finite-private-memory near-optimal strategies in the Big Match, it remains an open problem whether such strategies exist in all stochastic games.

\paragraph{\bf Connection to Extensive Form Correlated Equilibrium.}
In non-zero-sum stochastic games---which are outside the scope of this paper---public-memory processes enable players to coordinate their actions over time, facilitating reciprocity and potentially leading to more cooperative outcomes. This is reminiscent of the role of public signals in extensive-form correlated equilibria, where players condition their strategies on shared information to achieve higher payoffs. However, in zero-sum settings, public memory does not provide a coordination advantage but instead gives the opponent additional information, expanding their strategic options.

\section{The main result}

\begin{thm}\label{maintheorem}
Let $\Gamma = \langle Z, I, J, r, q \rangle$ be any stochastic game. For every $\varepsilon > 0$, there exist:
\begin{itemize}
    \item a memory process $(m_t)$ with \textbf{stationary probabilistic} memory updating,
    \item an $(m_t)$-based strategy $\sigma$ of Player~1 with \textbf{time-independent action selection},
    \item constants $K_\varepsilon = O(1/\varepsilon)$ and $n_\varepsilon > 0$,
\end{itemize}
such that for \textbf{every} $(\overline{m}_t)$-based strategy $\tau$ of Player~2, the following hold:
\begin{enumerate}
    \item[(a)] Uniform $\varepsilon$-optimality:
    \begin{equation}\label{ine:uniforme}
    \gamma_n(\sigma, \tau) := E_{\sigma, \tau} \overline{r}_n \geq v(z_1) - \varepsilon \quad \forall n \geq n_\varepsilon.
    \end{equation}
    \item[(b)] High-probability memory bound:
    \begin{equation}\label{ine:probmax}
    P_{\sigma, \tau}\left( \max_{t \leq n} m_t \geq K_\varepsilon \log n \right) \leq n^{-2}.
    \end{equation}
    \item[(c)] Almost-sure asymptotic memory bound:
    \begin{equation}\label{ine:problimsup}
    \limsup_{n \to \infty} \frac{\max_{t \leq n} m_t}{\log n} \leq K_\varepsilon \quad \text{almost surely under } P_{\sigma, \tau}.
    \end{equation}
    \item[(d)] High-probability uniform bound:
    \begin{equation}\label{ine:highprobmax}
    P_{\sigma, \tau} \left( \exists n : m_n > n_\varepsilon + K_\varepsilon \log n \right) < \varepsilon.
    \end{equation}
\end{enumerate}
\end{thm}

Note that $n_\varepsilon + K_\varepsilon \log n = O(\log n)$ for fixed $\varepsilon$. Since memory states are indexed by nonnegative integers, the number of distinct memory states used in the first $n$ stages is at most $1 + \max_{t \leq n} m_t$. Thus, inequality~\eqref{ine:highprobmax} implies the existence of a function $f(n) = O(\log n)$ such that, with probability at least $1 - \varepsilon$, the number of memory states used in the first $n$ stages never exceeds $f(n)$.

Allowing the memory updating to be {\em time-dependent} enables a uniform $\ve$-optimal strategy that uses no more than $1+K_\ve\ln n$ memory states in the first $n$ stages. This leads to the following result:

\begin{cor}[Time-dependent memory updating]%\label{cor:timedependent}
For every stochastic game $\Gamma=\langle Z, I,J, r, q\rangle$
and every positive number $\ve>0$, there is a memory-process $(m_t)$ with time-dependent memory updating and an $(m_t)$-based strategy $\s$ with time-independent choice of actions and positive numbers  $K_{\ve}=O(\frac{1}{\ve})$ and $n_{\ve}$ 
{such that inequality (\ref{ine:uniforme}) holds for every $(\overline{m}_t)$-based strategy $\tau$ of Player~2, and}

\begin{equation}\label{ine:memorymax}
m_n\leq K_\ve \ln n \;\;\;\forall n\geq 1.
\end{equation}
\end{cor}

Our strategy, like the Mertens–Neyman near-optimal strategy, instructs the player to act at each stage as if the discount rate were fixed, while dynamically adjusting this rate based on previous outcomes.

In our construction, the memory counter takes values in the set ${\gamma^i M : i \in \mathbb{N}}$, where $M$ is a sufficiently large constant and $\gamma > 1$ is a parameter depending on $\varepsilon$, chosen close to 1. At each stage, the counter is updated probabilistically based on the current stage outcome, with only a small chance of increase or decrease. This stochastic update rule is calibrated to match, in expectation, the update used in the Mertens–Neyman construction. However, unlike other possible stochastic memory schemes that could reduce memory usage even further, our construction is carefully tuned to preserve uniform $\varepsilon$-optimality. This balance is achieved by coupling the memory updates with a carefully designed function that maps counter values to discount rates. The full construction and analysis are presented in Section~\ref{sec:theproof}.

This approach opens the door to future refinements of memory-efficient strategies that trade off between probabilistic memory control and performance guarantees.

\section{The proof of the main result}
\label{sec:theproof}

Let $\Gamma=\langle Z, I,J, r, q\rangle$ be a stochastic game. Without loss of generality, assume that  $0\leq r(z,a)\leq 1$ for all states $z$ and action pairs $a$.

The $\lambda$-discounted game is the stochastic game where the total payoff is $\sum_{i=1}^\infty \lambda (1-\lambda)^{i-1}x_i$, where $x_i$ is the payoff at stage $i$, i.e., $x_i=r(z_i,a_i)$, where $z_i$ is the state at stage $i$ and $a_i$ is the action pair at stage $i$.

The value of  the $\lambda$-discounted game, as a function of the initial state $z$, exists \cite{sha53} and is denoted by $v_\lambda=(v_\lambda (z))_{z\in Z}$. Each player has, for each $0<\lambda<1$, a stationary strategy that is optimal in the $\lambda$-discounted game.

\textbf{Definition of the memory process:}
We define a memory-based strategy $\s=\s_{\ve,M,\lambda}$, which depends on three components: a precision parameter $\ve\in (0,1)$, a threshold $M>2$, and a discount-rate function $\lambda:(1,\infty)\to (0,1)$. 

We will prove that $\s_{\ve,M,\lambda}$, hereafter simply $\s$, satisfies the uniform $\ve$-optimality bound~\eqref{ine:uniforme}, as stated in Theorem~\ref{maintheorem}.  

Set $\gamma=\gamma_\ve=1+\ve/9$, and note that $\ln \gamma_\ve=O(\ve)$. The set of memories, which is the set $\mathbb{N}$ of nonnegative integers, is identified with the set of nonnegative numbers $S\!:=\{\gamma^k M: k \in \mathbb{N}\}$ via the bijective  map $k\mapsto \gamma^k M$.

Let $m_i$ be the memory at the beginning of stage $i$, and let $m_1=0$. Set $s_i=\gamma^{m_i}M$ and $\lambda_i=\lambda(s_i)$.

The stationary memory process $(m_t)$ evolves adaptively, increasing or decreasing based on the deviation of the observed payoff from the discounted game value. The memory updating is stochastic. Unlike earlier near-optimal strategies that use deterministic memory updating, our approach leverages stochastic memory updating, which plays a crucial role in reducing the number of memory states used.

The stationary memory process $(m_t)$ is such that
\[s_{i+1}\!\in \!\{\gamma s_i, s_i, \gamma^{-1}s_i\}\cap S,\]
\[s_{i+1}\geq s_i\; \text{ whenever }\; x_i-v_{_{\lambda_i}}(z_{i+1})+{\ve/2}\geq 0,\;
\text{ and }\]
\[s_{i+1}\leq s_i\; \text{ whenever }\; x_i-v_{_{\lambda_i}}(z_{i+1})+{\ve/2}\leq 0.\]

The stochastic law of the memory process $(m_t)$ is defined by the conditional probability of $s_{i+1}$ given $z_1,s_1,a_1,\ldots,z_i,s_i,a_i,z_{i+1}$, which is a function of only the triple $(s_i,x_i=r(z_i,a_i),z_{i+1})$.
\[P_{\s} (s_{i+1}=\gamma s_i\mid s_i,x_i,z_{i+1})=\frac{x_i-v_{_{\lambda_i}}(z_{i+1})+{\ve/2}}{s_i(\gamma -1)} \cdot 1_{\{x_i-v_{_{\lambda_i}}(z_{i+1})+{\ve/2}>0\}},\]

\[P_{\s} (s_{i+1}= \gamma^{-1}s_i \mid s_i,x_i,z_{i+1})=\frac{x_i-v_{_{\lambda_i}}(z_{i+1})+{\ve/2}}{s_i(\gamma^{-1} -1)}
\cdot 1_{\{ x_i-v_{_{\lambda_i}}(z_{i+1})+{\ve/2}<0\}}\cdot 1_{\{s_i>M\}}, \]
and (therefore)
\begin{equation*}
\begin{split}
P_{\s} (s_{i+1}= s_i \mid s_i,x_i,z_{i+1})&= 1- P_{\s} (s_{i+1}=\gamma s_i\mid s_i,x_i,z_{i+1})\\ &\quad - P_{\s} (s_{i+1}= \gamma^{-1}s_i \mid s_i,x_i,z_{i+1}).
\end{split}
\end{equation*}

This completes the definition of the memory process $(m_t)$, which, by construction, is a stationary memory process.

\begin{tcolorbox}[title=Probabilistic Memory Updating Reduces Memory Usage]
The probabilistic memory updating controls memory growth by ensuring that, at each stage, there is only a small probability of transitioning to a new, previously unused memory state. While alternative stochastic update rules can reduce memory usage even further, our construction is carefully tuned to achieve two important goals: logarithmic memory usage \emph{and} uniform $\varepsilon$-optimality. This balance is made possible by coupling the memory updates with a function that maps counter values to discount rates in a tightly calibrated way.

Conceptually, the probabilistic memory updating that enables reduced memory usage
resembles the classic technique of approximate counting~\cite{M78,F85}, albeit with notable differences. One main such difference is that our memory states can both increase and decrease.
\end{tcolorbox}

\textbf{The definition of the memory-based strategy:}
The $(m_t)$-based strategy $\s$ plays at stage $i$ an optimal strategy in the $\lambda_i:=\lambda(s_i)$-discounted game.
This completes the definition of the $(m_t)$-based strategy $\s_{\ve,M,\lambda}$.
By construction, the choice of actions of the $(m_t)$-based strategy $\s_{\ve,M,\lambda}$ is time independent.

\textbf{Bounding memory usage:}
First, we show that the strategy $\s$ uses a small number of memory states.

Fix an $(\overline{m}_i)$-based strategy $\tau$ of player 2. Lemma 1 below shows that inequality (\ref{ine:probmax}) holds, and Lemma 2 below shows that inequality
(\ref{ine:problimsup}) holds.

The probability distribution that is defined by $\s$ and $\tau$ on plays and memories is denoted by $P_{\s,\tau}$, or $P$ for short. The expectation w.r.t. $P_{\s,\tau}$ is denoted by $E_{\s,\tau}$, or $E$ for short.

%%%%%%%%%%%%%%%%%%%%%%%%%%%%%%%%%%%%%%%%%%%%%%%%%%%%%%%%%%%%%%%%%%%%%%%%%%
Let $C$ be a sufficiently large constant and let $K_\ve$ be such that
\[\frac{C}{\ln \gamma}\geq K_\ve \geq \frac{4}{\ln \gamma}.\]
It follows that $K_\ve=O(1/\ve)$.
%%%%%%%%%%%%%%%%%%%%%%%%%%%%%%%%%%%%%%%%%%%%%%%%%%%

\begin{lem}For every $(\overline{m}_i)$-based strategy $\tau$ of Player 2 and for
all integers $n\geq M>2$,
\[P_{\s,\tau}(\max_{i=1}^n m_i\geq K_\ve \ln n)\leq \frac{1}{n^2}.\]

\end{lem}
\begin{proof}
The stochastic law of $s_i$ guarantees that
\[s_{i+1}-s_i>0 \implies s_{i+1}-s_i=s_i(\gamma -1)\]
and (using the inequality $x_i-v_{\lambda_i}(z_i)+\ve/2\leq 2$)
\[P_{\s,\tau}(s_{i+1}-s_i>0\mid z_1,s_1,i_1,j_1,\ldots, z_i,s_i)\leq \frac{2}{s_i(\gamma-1)}.\]
Therefore,
\[E_{\s,\tau}(s_{i+1}-s_i\mid z_1,s_1,i_1,j_1,\ldots, z_i,s_i)\leq 2.\]

Therefore, as the expectation equals the expectation of the conditional expectation, $E_{\s,\tau}(s_{i+1}-s_i)\leq 2$. Therefore,  $E_{\s,\tau}s_{i+1}\leq 2i+M$. The random variable $s_i$ is nonnegative. Therefore, by Markov's inequality,
\[P_{\s,\tau}(s_i\geq \gamma^k M)\leq \frac{E_{\s,\tau}s_i}{\gamma^k M}\leq \frac{2(i-1)+M}{\gamma^k M}.\]
Therefore, for every positive integer $k$,
%%%%%%%%%%%%%%%%%%%%%%%%%%%%%%%%%%%%%%%%%%%%%%%%%%%%%%%%%%
\begin{eqnarray*}P_{\s,\tau}(\max_{i=1}^n m_i\geq k)
&=&P_{\s,\tau}(\max_{i=1}^n s_i\geq \gamma^kM)
= P_{\s,\tau}(\exists i\leq n \text{ s.t. } s_i\geq \gamma^k M)\\
&\leq &\sum_{i=1}^nP_{\s,\tau}(s_i\geq \gamma^kM) \leq \sum_{i=1}^n\frac{2(i-1)+M}{\gamma^k M}\\
&=&\frac{n^2-n+nM}{\gamma^k M}.
\end{eqnarray*}
%%%%%%%%%%%%%%%%%%%%%%%%%%%%%%%%%
Therefore, for $n\geq M>2$, we have
\begin{equation*}
P_{\s,\tau}(\max_{i=1}^n m_i\geq k)\leq \frac{2n^2}{\gamma^kM}\leq \frac{n^2}{\gamma^k}.
\end{equation*}
%%%%%%%%%%%%%%%%%%%%%%%%%%%%%%%%%%%%%%%%%%%%%%%%%%%%%%%%%%%
Hence, by letting $k_n$ be the smallest integer that is $\geq K_\ve\ln n$, for all $n\geq M>2$, we have
\begin{eqnarray*}P_{\s,\tau}(\max_{i\leq n}m_i\geq K_\ve\ln n)&=&P_{\s,\tau}(\max_{i\leq n}m_i\geq k_n)\leq \frac{n^2}{\gamma^{k_n}}\\ &\leq& \frac{n^2}{\gamma^{K_\ve\ln n}}\leq  \frac{n^2}{ e^{\ln \gamma(4/\ln \gamma)\ln n }}= \frac1{n^2}.
\end{eqnarray*}
\end{proof}

\begin{lem}%For $K_\ve=\frac{4}{\ln \gamma}$, we have %$K_\ve=O(\frac{1}{\ve})$ and
\[P_{\s,\tau}\left(\limsup_{n\to\infty}\frac{\max_{i\leq n}m_i}{\ln n}\leq K_{\ve}\right)=1.\]
\end{lem}
\begin{proof} As $P_{\s,\tau}(\max_{i\leq n}m_i\geq K_\ve \ln n)\leq P_{\s,\tau}(\max_{i\leq n}m_i\geq \frac{4\ln n}{\ln \gamma})\leq n^{-2}$, and as $\sum_n n^{-2}<\infty$, the sum of probabilities converges. By the Borel-Cantelli lemma, almost surely under  $P_{\s,\tau}$, only finitely many values of $n$ satisfy $\frac{\max_{i\leq n}m_i}{\ln n}\geq K_\ve$. Consequently, we have $P_{\s,\tau}(\limsup_{n\to \infty}\frac{\max_{i\leq n}m_i}{\ln n}\leq K_\ve)=1$.
\end{proof}

\textbf{The map $\lambda$ from memories to discount rates:}
Fix $0<\ve<1/4$ and recall that $1<\gamma=\gamma_\ve=1+\ve/9$. The sufficiently large constant $M$ will be defined in the sequel.
Define the function $\lambda: (1,\infty)\to \mathbb{R}_+$ by
\[\lambda(s) =\frac{1}{s\ln^2 s}=\frac{1}{s(\ln s)^2}.\]

\begin{tcolorbox}[title=Choice of the Discount Rate Function $\lambda(s)$: Balancing Slow Decay and Integrability]
This choice ensures that the discount rate decays slowly enough to bound the differences \( v_{\lambda_{i+1}}(z) - v_{\lambda_i}(z) \) by a small multiple of \( \lambda_i \) (see inequality~\eqref{ineq:vlambdachange}), yet fast enough to ensure that the function \( s \mapsto \lambda(s) \) is integrable.

Any function of the form \( \lambda(s) = \frac{1}{s \log^{1+\eta} s} \) with \( \eta > 0 \), or even \( \lambda(s) = \frac{1}{s \log s \log \log^{1+\eta} s} \), satisfies these requirements and is independent of the underlying stochastic game. For each fixed stochastic game, one may also take \( \lambda(s) = \frac{1}{s^{1+\eta}} \), where \( \eta \) depends on the game’s structure.
\end{tcolorbox}

%%%%%%%%%%%%%%%%%%%%%%%%%%%%%%%

%%%%%%%%%%%%%%%%%%%%%%%%%%%%%%%%%%%%%%

%%%%%%%%%%%%%%%%%%%%%%%%%%%%%%%%%%%%%%%

%%%%%%%%%%%%%%%%%%%%%%%%%%%%%%%%%

%%%%%%%%%%%%%%%%%%%%%%%%%%%%%%%%%%%

\textbf{Properties of the functions $s\mapsto\lambda$ and $\lambda\mapsto v_\lambda$:}
First, we list a few properties of the function $\lambda$. The function $\lambda$ is the derivative of the function $-1/\ln s$. Therefore,
\begin{equation}\label{eqn:slambdaintegration}
\frac{1}{\ln s}-\frac{1}{\ln s'}=\int_s^{s'}\lambda(s)\,ds.
\end{equation}
The function $\lambda$ is differentiable and its derivative at $s$ equals $-\lambda^2(s)(\ln^2 s+2\ln s)$. Therefore, using the inequality $2\ln s\leq \ln ^2 s\; \forall s\geq e^2$, we have
\begin{equation}\label{ine:dlambdadsbound}
|\frac{d\lambda}{ds}(s)|=\lambda^2(s)(2\ln s+\ln^2 s)\leq 2\lambda^2(s)\ln^2 s \;\; \forall s\geq e^2.
\end{equation}

Second, we derive a few properties of the function $\lambda\mapsto v_\lambda$.
The limit of $v_\lambda$ as $\lambda\to 0+$ exists by the result of Bewley and Kohlberg~\cite{bewkoh76}, and is denoted by $v$. The assumption that  $0\leq r\leq 1$ implies that $0\leq v_\lambda \leq 1$ and thus also $0\leq v \leq 1$. For $u\in \mathbb{R}^Z$, $\max_{z\in Z}|u(z)|$ is denoted by $\|u\|$.

The expansion, due to~\cite{bewkoh76}, of $v_\lambda$ as a convergent series in fractional powers of $\lambda$, implies the existence of positive numbers  $1>\lambda_0>0$, $K>2$, and $1\geq \beta>0$ such that $v_\lambda$ is differentiable in the interval  $(0,\lambda_0)$ and $\|\frac{dv_{\lambda}}{d\lambda}\|\leq K\lambda^{\beta-1}$ for every  $0<\lambda<\lambda_0$.
W.l.o.g. we assume that $\lambda_0< 1/K$. Hence,
\begin{equation}\label{ine:dvdlambdabound}\|\frac{dv_{\lambda}}{d\lambda}\|\leq \lambda^{\beta-1}/\lambda_0\;\;  \forall\;0<\lambda<\lambda_0.
\end{equation}

Fix such positive numbers  $1>\lambda_0>0$ and $1\geq\beta>0$.

We will establish inequalities (\ref{eqn:n2}), (\ref{eqn:vlambdav}), (\ref{eqn:Mlarge enotgh1}), and (\ref{eqn:n1}), which are used in proving that $\s_{\ve,M,\lambda}$ is uniform $\ve$-optimal.

%%%%%%%%%%%%%%%%%%%%%%%%%%%%%%%%%%%%%%%%
The next result bounds the variation of the function $s\mapsto v_{\lambda(s)}$.

\begin{lem}\label{lem:variationv}There is a positive constant $M_1$ such that for all $s'\geq s\geq M_1$
\begin{eqnarray}\label{eqn:n2}
\|v_{\lambda(s)}-v_{\lambda(s')}\|&\leq &\frac{\ve (\gamma-1)}{\ln s}-\frac{\ve (\gamma-1)}{\ln s'}
= \frac{\ve^2/9}{\ln s}-\frac{\ve^2/9}{\ln s'}.
%\\ \nonumber &=&O(\frac{1}{\ln s}-\frac{1}{\ln s'})~.
\end{eqnarray}
\end{lem}
\begin{proof}
%%%%%%%%%%%%%%%%%%%%%%%%%%%%%%%%%%%%%%%%%%%%%%%%%%%%%%%%%%%%
As 
$\beta>0$, $\frac{2(\ln s)^{2}}{s^\beta(\ln s)^{2\beta} }\to_{s\to\infty} 0$. This follows since the denominator grows faster than the numerator for any fixed $\beta>0$.
Let $M_1>\lambda_0^{-1}$ be a sufficiently large positive constant such that $M_1>e^2$ and
\begin{equation}\label{ineq:star1}\frac{2(\ln s)^{2}}{s^\beta(\ln s)^{2\beta} }<{\lambda_0}\ve(\gamma-1)\;\; \forall s\geq M_1.
\end{equation}

For $s\geq M_1$, $\lambda(s)<\lambda_0$ and $s>e^2$. Therefore, inequalities (\ref{ine:dvdlambdabound}), (\ref{ine:dlambdadsbound}),  and (\ref{ineq:star1}),
%the bound $\|\frac{dv_{\lambda}}{d\lambda}\|\leq \lambda^{\beta-1}/{\lambda_0}$ for $0<\lambda<\lambda_0$ together with the inequality %$|\frac{d\lambda}{ds}(s)|<2\lambda^2(s)(\ln s)^2$
imply that

\begin{eqnarray*}\|\frac{dv_{\lambda (s)}}{ds}\|&\leq &\frac{\lambda^{\beta-1}(s)}{{\lambda_0}}|\frac{d\lambda}{ds}(s)|\leq 2\lambda(s)\lambda^{\beta}(s)(\ln s )^2/{\lambda_0}\\ &= & %\frac{2\lambda(s)^\beta}{ s}=\frac{2}{s s^\beta (\ln s)^{2\beta}}=
\frac{2(\ln s)^{2}}{s^\beta (\ln s)^{2\beta} }\lambda(s)/\lambda_0<{\lambda_0}\ve(\gamma-1)\lambda(s)/{\lambda_0}\\
&=& \ve(\gamma -1)\lambda (s).
\end{eqnarray*}
Therefore, for $s'\geq s\geq M_1$, we have,
\begin{eqnarray*}\|v_{\lambda(s)}-v_{\lambda(s')}\|&\leq &\int_s^{s'}\|\frac{dv_{\lambda(s)}}{ds}\|\,ds\\
&\leq & \ve(\gamma-1)\int_s^{s'}\lambda(s)\,ds=\frac{\ve (\gamma-1)}{\ln s}-\frac{\ve (\gamma-1)}{\ln s'},
\end{eqnarray*}
which completes the proof of the lemma.
\end{proof}

%%%%%%%%%%%%%%%%%%%%%%%%%%%%%%%%%%%%%%%%%%%%%%%%%%%%%%%%%%%
We continue with the derivation of a few inequalities of various functions of $s$. Recall that the function $s\mapsto \frac{1}{s\ln^2 s}=\lambda(s)$ is monotonically decreasing and the limit of $v_{\lambda}$, as ${\lambda\to 0+}$, exists and equals $v$.
As the function $s\mapsto \frac{1}{\ln s}$ decreases to $0$ as $s\to\infty$,
there is a positive constant $M_2$ such  that
\begin{equation}\label{eqn:vlambdav} v_{\lambda(s)}(z)\geq v(z)-\ve/8 +\frac{1}{\ln M_2}\;\;\; \forall z\in Z \text{ and } s\geq M_2.
\end{equation}

Recall that $0<\ve<1/4$ and that $\gamma=1+\ve/9$. Therefore, $\ln \gamma<\frac{1}{36}<2^{-5}$.
As $\frac{\ln \gamma \ln (\gamma s)}{\ln s}\to_{s\to\infty}\ln \gamma$, there is a sufficiently large $M_3$ such that
\begin{equation}\label{eqn:Mlarge enotgh1}\frac{\ln \gamma \ln \gamma s}{\ln s}<2^{-5}\;\; \forall s\geq M_3.\end{equation}

The definition of $\lambda(s)$ implies that
$\lambda(\gamma s)/\lambda(s)\to_{s\to\infty}\gamma^{-1}>1-\ve/9$
and
$\lambda(\gamma^{-1}s)/\lambda(s)\to_{s\to\infty}\gamma=1+\ve/9$.
Along the monotonicity of $s\mapsto \lambda (s)$ we deduce that there is a constant $M_4$ such that
\begin{equation}\label{eqn:n1}
|\lambda(s)-\lambda(s')|<{\ve} \lambda(s)/8 \;\;\; \forall s ,s'\geq M_4 \mbox{ with } \gamma^{-1} s\leq s'\leq \gamma s.
\end{equation}
%%%%%%%%%%%%%%%%%%%%%%%%%%%%%%%%%%%%%%%%%%%%%%%%%%%%%%

%%%%%%%%%%%%%%%%%%%%%%%%%%%%%%%%
Equality (\ref{eqn:slambdaintegration}) along with inequality (\ref{eqn:n1}), imply that for $M>M_4$,
\begin{equation}\label{ineq:fromlnsdifferencetosdifferences}
\frac{1}{\ln s}\!-\!\frac{1}{\ln s'}\geq \lambda(s)(s'\!-\!s-\ve |s'\!-\!s|/8) \;\; \forall s ,s'\geq M \mbox{ with } \gamma^{-1} s\leq s'\leq \gamma s.
\end{equation}

%%%%%%%%%%%%%%%%%%%%%%%%%%%%%%%%%%%%%%%%%%%%%%%%%%%%%%%%%%
\textbf{Bounding the payoff from below:}
Now, we will prove that for $M>\max (M_1,M_2,M_3,M_4)$, where:
\begin{itemize}
  \item $M_1$ ensures that the difference  $v_{\lambda_{i+1}}-v_{\lambda_i}$ is small (see Lemma~\ref{lem:variationv}),
  \item $M_2$ guarantees that $v_{\lambda(s)}(z) \geq v(z) - \varepsilon/8$ (see~\eqref{eqn:vlambdav}),
  \item $M_3$ ensures that $\frac{\ln \gamma \ln \gamma s}{\ln s}$ is sufficiently small (see~\eqref{eqn:Mlarge enotgh1}),
  \item $M_4$ guarantees that 
$|\lambda(s)-\lambda(s')|<{\ve} \lambda(s)/8$ whenever $s'\in[\gamma^{-1}s,\gamma s]$ (see~\eqref{eqn:n1}),
\end{itemize} 
the strategy $\s=\s_{\ve,M,\lambda}$ obeys inequality~\eqref{ine:uniforme}.

Let $\mcF_i$ denote the algebra of the play up to stage $i$, including the sequence of memories $s_1,\ldots,s_i$ and the state $z_i$.

Recall that $\ve<1/4$ and $0\leq r\leq 1$. Therefore, $|x_i-v_{_{\lambda_i}}(z_{i+1})+\ve/2|\leq  1+\ve/2<9/8$. Therefore,
the definition of the conditional probabilities of $s_{i+1}$, given $s_i,x_i,z_{i+1}$ implies that for every $(\overline{m}_t)$-based strategy $\tau$ of Player 2, we have
\begin{eqnarray}E_{\s,\tau} (|s_{i+1}- s_i|\mid \mcF_i)&\leq & E_{\s} (|x_i-v_{_{\lambda_i}}(z_{i+1})+\ve/2| \mid s_i,x_i,z_{i+1})<9/8.
\end{eqnarray}

The definition of the conditional probabilistic law of $s_{i+1}$ has three implications.
First,
(by using $\ve<1/4$ and therefore $-1<x_i-v_{_{\lambda_i}}(z_{i+1})+{\ve/2}<2$)%, and $\gamma-1<1/3$ and %therefore $1-\gamma^{-1}<1/3$)
\begin{equation}\label{eqn:n3}
P(s_{i+1}\neq s_i \mid \mcF_i)\leq \frac{2}{s_i(\gamma-1)}. % \mbox{ and } E((s_{i+1}- s_i)_+ \mid \mcF_i)\leq 2,
\end{equation}

 Second,
as $E(x_i-v_{_{\lambda_i}}(z_{i+1})+{\ve/2}\mid \mcF_i)=E(s_{i+1}-s_i\mid \mcF_i)$ on $s_i>M$ and  $E({x_i-v_{_{\lambda_i}}(z_{i+1})} +  {\ve/2}\mid \mcF_i)  \geq E(s_{i+1}-s_i\mid \mcF_i)-1$ on $s_i=M$,
\begin{equation}\label{eqn:3n}E(x_i-v_{_{\lambda_i}}(z_{i+1}) \mid \mcF_i)\geq -{\ve/2}+E (s_{i+1}-s_i\mid \mcF_i)-1_{\{s_i=M\}}.
\end{equation}

As $\s$ plays at stage $i$ an optimal strategy in the $\lambda_i$-discounted game,
$E(\lambda_i x_i +(1-\lambda_i)\,v_{_{\lambda_i}} (z_{i+1}) \mid \mcF_i)\geq v_{_{\lambda_i}} (z_{i})$, and therefore,
\begin{equation}\label{eqn:1}
E(v_{_{\lambda_i}} (z_{i+1})-v_{_{\lambda_i}} (z_{i})+\lambda_i (x_i-v_{_{\lambda_i}}(z_{i+1})) \mid \mcF_i)\geq 0.
\end{equation}

By (\ref{eqn:n2}) and (\ref{eqn:n3}),
\begin{eqnarray*}E\left(\|v_{_{\lambda_{i+1}}}-v_{_{\lambda_{i}}}\|\mid \mcF_i\right)&\leq & E(\frac{2}{s_i(\gamma-1)}\;\left|\frac{\ve(\gamma -1)}{\ln s_i}-\frac{\ve(\gamma -1)}{\ln s_{i+1}}\right|\mid \mcF_i)\\
&=& E\left(\frac{2\ve |\ln s_{i+1}-\ln s_i|}{s_i\ln s_i \ln s_{i+1}}\mid \mcF_i\right)
\leq E\left(\frac{2\ve \ln \gamma}{s_i\ln s_i \ln s_{i+1}}\mid \mcF_i\right)\\
&=&E\left(\frac{2\ve \ln \gamma \ln s_i}{s_i\ln^2\!s_i \ln s_{i+1}}\mid \mcF_i\right)=2\ve \lambda_iE\left(\frac{\ln \gamma \ln s_i}{ \ln s_{i+1}}\mid \mcF_i\right).
\end{eqnarray*}
  Therefore, by using inequality (\ref {eqn:Mlarge enotgh1}), %$| \frac{\ln \gamma \ln s_i}{\ln s_{i+1}}|<2^{-4}$,
  we deduce that
\begin{equation}\label{ineq:vlambdachange} E(v_{_{\lambda_{i}}}(z_{i+1})-v_{_{\lambda_{i+1}}}(z_{i+1})\mid \mcF_i)\geq -\ve \lambda_i/16.\end{equation}

By adding inequalities (\ref{eqn:3n}) and (\ref{ineq:vlambdachange}), we have
\begin{eqnarray*}E(x_i-v_{_{\lambda_{i+1}}}(z_{i+1}) \mid \mcF_i)&\geq& -\ve/2 +E (s_{i+1}-s_i\mid \mcF_i) - 1_{\{s_i=M\}}-\ve \lambda_i/16\\ &\geq&  -9\ve/16 +E (s_{i+1}-s_i\mid \mcF_i) - 1_{\{s_i=M\}}.\end{eqnarray*}

As the expectation is the expectation of the conditional expectation, we deduce that
\[E\,x_i \geq E\,v_{_{\lambda_{i+1}}}(z_{i+1}) -9\ve/16 +E (s_{i+1}-s_i) - E\,1_{\{s_i=M\}}.\]

Summing these inequalities over $i=1,\ldots,n$, using the inequality $s_{n+1}-s_1\geq 0$, and dividing by $n$, we deduce that
\begin{equation}\label{eqn:preparetouniform}E\frac1n\sum_{i=1}^nx_i \geq E\frac1n\sum_{i=1}^nv_{_{\lambda_{i+1}}}(z_{i+1})-9\ve/16 - E\frac1n\sum_{i=1}^n1_{\{s_i=M\}}.
\end{equation}

\begin{lem}\label{lem:3}
\begin{eqnarray}\;\;\label{eqn:15}E\frac1n\sum_{i=1}^nv_{_{\lambda_{i+1}}}(z_{i+1})&\geq &v(z_1)-\ve/8,
\text{ and }\\
\label{eqn:16}\;\;\;\;\;\;-E\frac1n\sum_{i=1}^n1_{\{s_i=M\}}&\geq & \frac{-9}{n\ve \lambda(M)}\;\geq\; -\ve/8 \hspace{1.8cm} \forall n\geq \frac{72}{\ve^2 \lambda(M)}.
\end{eqnarray}
\end{lem}

Before proving the lemma, we show that the lemma along with inequality (\ref{eqn:preparetouniform}) shows that $\s$  satisfies inequality (\ref{ine:uniforme}). Indeed, summing inequalities
(\ref{eqn:preparetouniform}), (\ref{eqn:15}), and (\ref{eqn:16}), and cancelling terms that appear in both sides of the sum of the inequalities, we have that
\[E\frac1n\sum_{i=1}^nx_i \geq v(z_1)-\frac{9\ve}{16}-\frac{\ve}{8}-\frac{\ve}{8}>v(z_1)-\ve \hspace{2cm} \forall n\geq \frac{72}{\ve^2 \lambda(M)},\]
which proves that $\s$ satisfies inequality (\ref{ine:uniforme}) with $n_\ve= \frac{72}{\ve^2 \lambda(M)}$.

\smallskip
Now we turn to the proof of Lemma \ref{lem:3}.
\begin{proof}
Define $Y_i=v_{_{\lambda_i}}(z_i)-\frac{1}{\ln s_i}$. Recall that we write $E$ for $E_{\s,\tau}$ for short.

In the following chain of an equality and inequalities,
equality (\ref{ineq:Ytovs}) follows from the definition of $Y_i$; inequality (\ref{ineq:tosdifferences}) follows by adding to the right hand side of equality (\ref{ineq:Ytovs}) inequality (\ref{ineq:vlambdachange}); inequality (\ref{ineq:tolambdas}) follows from inequality (\ref{ineq:fromlnsdifferencetosdifferences}); inequality (\ref{ineq:tolambdasaddition}) follows from the inequality $E(|s_{i+1}-s_i|\mid \mcF_i)\leq 2$;
inequality (\ref{ineq:tolambdaxv}) follows from the definition of the conditional distribution of $s_{i+1}$ given $(s_i, x_i, z_{i+1})$; and inequality (\ref{ineq:usingbasicineq}) follows from inequality (\ref{eqn:1}).

\begin{eqnarray}\nonumber  & &E(Y_{i+1}-Y_i\mid \mcF_i)\\
\label{ineq:Ytovs}&= &E(v_{_{\lambda_{i+1}}}(z_{i+1})-v_{_{\lambda_i}}(z_i)+\frac{1}{\ln s_i}-\frac{1}{\ln s_{i+1}}\mid \mcF_i)\\
\label{ineq:tosdifferences}&\geq & E(v_{_{\lambda_{i}}}(z_{i+1})-v_{_{\lambda_i}}(z_i)+\frac{1}{\ln s_i}-\frac{1}{\ln s_{i+1}}\mid \mcF_i)-\ve\lambda_i/16\\
\label{ineq:tolambdas}     &\geq & E(v_{_{\lambda_{i}}}(z_{i+1})-v_{_{\lambda_i}}(z_i)+\lambda_i(s_{i+1}-s_i-\ve|s_{i+1}-s_i|/8)\mid \mcF_i)\\
\nonumber & & -\ve\lambda_i/16\\
\label{ineq:tolambdasaddition}
&\geq & E(v_{_{\lambda_{i}}}(z_{i+1})-v_{_{\lambda_i}}(z_i)+\lambda_i(s_{i+1}-s_i)\mid \mcF_i)-5\ve\lambda_i/16\\
\label{ineq:tolambdaxv}    &\geq & E(v_{_{\lambda_{i}}}(z_{i+1})-v_{_{\lambda_i}}(z_i)+\lambda_i(x_i-v_{_{\lambda_i}}(z_{i+1}))\mid \mcF_i)+3\ve\lambda_i/16\\
\label{ineq:usingbasicineq}&\geq & 3\ve\lambda_i/16\geq \ve \lambda_i/8~.
\end{eqnarray}

By taking expectation we deduce that for every $j\geq 1$, we have $E Y_{j+1}-EY_j\geq 0$. Summing these inequalities over $j=1,\ldots,i$, we deduce that $EY_{i+1}\geq Y_1$.
As $v_{_{\lambda_{i+1}}}(z_{i+1})\geq Y_{i+1}$, which follow from the definition of $Y_i$,  and $Y_1\geq v(z_1)-\ve/8$, which follow from inequality (\ref{eqn:vlambdav}), we have, $E\,v_{_{\lambda_{i+1}}}(z_{i+1})\geq v(z_1)-\ve/8$ $\forall i\geq 1$, and hence inequality (\ref{eqn:15}) follows.

The above chain of inequalities shows that
\[E(Y_{i+1}-Y_i\mid \mcF_i)\geq \ve \lambda_i/8, \text{ and hence } EY_{i+1}-EY_i\geq E(\ve \lambda_i/8).\]
Summing these inequalities over $i=1,\ldots,n$ and using the inequalities $1 \geq Y_i$ and $Y_1\geq v(z_1)-\ve/8\geq -\ve/8$, we have
\[1\geq EY_{n+1}\geq Y_1+E\sum_{i=1}^n\ve \lambda_i/8\geq -\ve/8+E\sum_{i=1}^n\ve \lambda_i/8.\]
Therefore,
\begin{equation}\label{eqn:25}
9\geq E\sum_{i=1}^n\ve \lambda_i\geq \ve \lambda(M)E\sum_{i=1}^{n}1_{\{s_i=M\}}.
\end{equation}
From inequality (\ref{eqn:25}), we deduce that \[E\sum_{i=1}^{n}1_{\{s_i=M\}}\leq \frac{9}{\ve \lambda (M)}.\] Dividing by $n$, we obtain
\[E\frac1n\sum_{i=1}^{n}1_{\{s_i=M\}}\leq \frac{9}{n\ve \lambda(M)}\leq \frac{\ve}{8} \hspace{1cm} \forall n\geq \frac{72}{\ve^2\lambda(M)},\]
which establishes inequality (\ref{eqn:16}).
\end{proof}

%%% Local Variables:
%%% mode: latex
%%% TeX-master: "paper"
%%% End:

\section{The Big Match with a clock and a finite public memory}
Our result about the limitations of bounded public memory is shown for the Big Match, the influential example of a stochastic game introduced by Gillette~\cite{gil57} we described in the introduction. Recall that this game has a single nonabsorbing state, in which each player has two actions.

The two actions of player 1 are labeled $A$ (the absorbing action) and $C$ (the continuing and safe action). The two actions of player 2 are labeled 0 for the action with $r(C,0)=1$ (and thus $r(A,0)=0^*$, denoting that the game transitions to an absorbing state with payoff~0) and $1$ for the action with $r(C,1)=0$ (and thus $r(A,1)=1^*$, denoting that the game transitions to an absorbing state with payoff~1).

%%%%%%%%%%%%%%%%%%%%%%%%%%%%%%%%%%%%%%%%%%%

For a strategy pair $\s$ of player 1 and $\tau$ of player 2, let $\gamma_n(\s,\tau)$ denote the expected average payoffs to Player 1 over the first $n$ stages, under the the distribution induced by $\s$ and $\tau$:
\[\gamma_n(\s,\tau):=E_{\s,\tau}\frac1n\sum_{t=1}^nr_t.\]

\begin{thm}\label{thm:limsupEworthless}
For every positive integer $M$, memory process $(m_t)_{t=1}^\infty$ in ${\mathcal M}_M$, $\delta>0$, and an $(m_t)$-based strategy $\s$ of player 1, there is a strategy $\tau$ of player 2, such that
\begin{equation}\label{eqn:limsupEworthless}\limsup_{n\to\infty}\gamma_n(\s,\tau)\leq \delta.
\end{equation}
Moreover, such a strategy $\tau$ can be chosen as a mixture of finitely many, not necessarily distinct, pure $(m_t)$-based strategies, each selected with equal probability.
\end{thm}

\begin{corollary}\label{cor:liminfEworthless}For every positive integer $M$, memory process $(m_t)_{t=1}^\infty$ in ${\mathcal M}_M$, $\delta>0$, and an $(m_t)$-based strategy $\s$ of player 1, there is a pure $(m_t)$-based strategy $\tau$ of player 2, such that
\begin{equation}\label{eqn:liminfEworthless} \liminf_{n\to\infty}\gamma_n(\s,\tau)\leq \delta.
\end{equation}
\end{corollary}
\begin{proof}[Proof of Corollary \ref{cor:liminfEworthless}]
Let $\tau$ be the uniform mixture of the finitely many pure $(m_t)$-based strategies $\tau^i$, $i=1,\ldots,k$, of player 2 such that (\ref{eqn:limsupEworthless}) holds.
Then, $\frac1k\sum_{i=1}^k\gamma_n(\s,\tau^i)=\gamma_n(\s,\tau)$. As (\ref{eqn:limsupEworthless}) holds and
\[\frac1k\sum_{i=1}^k\liminf_{n\to\infty} \gamma_n(\s,\tau^i)\leq \liminf_{n\to\infty}\frac1k\sum_{i=1}^k \gamma_n(\s,\tau^i)\leq\limsup_{n\to\infty} \gamma_n(\s,\tau),\]
there is $i$ such that $\liminf_{n\to\infty}\gamma_n(\s,\tau^i)\leq \delta$.
\end{proof}

The proof of Theorem \ref{thm:limsupEworthless} is obtained by defining a sequence of (not necessarily distinct) $(m_t)$-based strategies $\tau^i$ of player 2,
such that (\ref{eqn:limsupEworthless}) holds for any strategy $\tau$ that is mixture of sufficiently many of the strategies $\tau^i$. The next lemma states the properties of the sequence of $(m_t)$-based strategies $\tau^i$ of player 2 that are used in the proof of (\ref{eqn:limsupEworthless}).

\begin{lem}\label{lem:Xstrategies}For every positive integer $M$ and ever memory process $(m_t)_{t=1}^\infty$ in ${\mathcal M}_M$, $\delta>0$,  and an $(m_t)$-based strategy $\s$ of player 1, there is a sequence $\tau^i$, $i\in \mathbb{N}$, of  pure $(m_t)$-based strategies of player 2 and a sequence $n_i$ of positive integers, such that
\begin{equation}\label{eqn:fewbadtauinew}
\forall t\geq n_i,\;\;\; \sum_{i=1}^\infty1_{\{r^i_t\geq \delta\}}1_{\{t\geq n_i\}}\leq M+1, \;\;\;\mbox{ where }\;r^i_t:=E_{\s,\tau^i}r_t.
\end{equation}
\end{lem}

\begin{proof}[Proof of Lemma \ref{lem:Xstrategies}]
Let $(m_t)_{t=1}^\infty$ be a memory process in ${\mathcal M}_M$, $M$ a positive integer, and let $\s$ be an $(m_t)$-based strategy of player 1.
We use the symbols $[M]$ %and $X$ also
to denote the sets $\{1,\ldots,M\}$. % and $\{1,\ldots X\}$ respectively.
A pure $(m_t)$-based strategy $\tau$ of player 2 is a function from $\mathbb{N}\times [M]$ to $\{0,1\}$. We identify the pure $(m_t)$-based strategy $\tau$ with the set $1_\tau$, which consists of all pairs $(t,m)$ such that $\tau(t,m)=1$. That is,
\[1_\tau:=\{(t,m): t\geq 1, 1\leq m\leq M, \mbox{ and } \tau(t,m)=1\}.\]

{\bf The properties of the strategies $\tau^i$.}

The set of strategies $\tau^i$, $1\leq i$, will satisfy the following properties.

\begin{equation}\label{property:monotonicity}
\tau^i(t,m)\leq \tau^{i+1}(t,m) \;\;\;\forall (t,m)\in \mathbb{N}\times [M],
\end{equation}
and
$\forall i\in \mathbb{N} \;\;\exists n_i \mbox{ s.t. } \forall t\geq n_i$,
\begin{equation}\label{property:existenceofdesiredm}
% E_{\s,\tau^i}r_t=
 r^i_t\geq \delta \implies \exists m \mbox{ s.t. } 0=\tau^i(t,m)<\tau^{i+1}(t,m)=1.
\end{equation}
%%%%%%%%%%%%%%%%%%%%%%%%%%%%%%%%%%%%%%%%

{\bf The definition of $\tau^i$.}
We define $\tau^i$ (as a function of the strategy $\s\in \mcM_M$) by induction on $i$. The definition will imply that
\begin{equation}\label{eqn:sumsigmataui}\sum_{(t,m)\in \tau^i}\s(t,m)[A]<\delta/3.
\end{equation}
%%%%%%%%%%%%%%%%%%%%%%%%%%%%

Set $T_*:=\inf \{t: i_t=A\}$, and if no absorbing action is played (i.e., if $\{t: i_t=A\}$ is the empty set), we set $T_*=\infty$.
$T_*$ is a stopping time whose distribution depends on the strategies  of player 1 and player 2. The event that the play of the game is absorbed at $1^*$, respectively at $0^*$, is denote by $1^*$, respectively $0^*$.

The expectation w.r.t. the probability $P_{\s,\tau}$ is denoted by $E_{\s,\tau}$.

Given a probability $P$ on plays, the $P$-probability of the event $1^*$, respectively $0^*$, is denoted by $P(1^*)$, respectively $P(0^*)$.
That is, ${P}(1^*):={P}(T_*=t<\infty, i_t=A, j_t={1})$ and ${P}(0^*)={P}(T_*=t<\infty, i_{t}=A, j_{t}=0)$.

Let ${P}(t,m)={P}(T_*\geq t \mbox{ and } m_t=m)$, and let $P_i$ denote the probability distribution induced by $\s$ and $\tau^i$, i.e., $P_i:=P_{\s,\tau^i}$.

Note that $P_i(1^*)=\sum_{(t,m)\in \tau^i}P_i(t,m)\s(t,m)[A]\leq \sum_{(t,m)\in \tau^i}\s(t,m)[A]$. Hence, if $\tau^i$ satisfies property (\ref{eqn:sumsigmataui}), then $P_i(1^*)<\delta/3$.

{\bf Definition of $\tau^1$.}
\[\tau^1(t,m)=0\;\; \forall (t,m).\]
Note that $\tau^1$ satisfies property (\ref{eqn:sumsigmataui}).

{\bf Inductive definition of $\tau^{i+1}$.}
Assume that $\tau^i$ is an $(m_t)$-based strategy (of player 2) that satisfies property (\ref{eqn:sumsigmataui}).
Set ${P}(1^*_{<t}):={P}(T_*=s<t, i_s=A, j_s=1)$ and ${P}(0^*_{>t}):={P}(T_*=s>t, i_s=A, j_s=0)$. Note that
\begin{eqnarray*}r^i_t&=&P_i(1^*_{\leq t})+\sum_{m: (t,m)\notin \tau^i} P_i (t,m)\sigma(t,m)[C]\\
&\leq &P_i(1^*_{\leq t})+\sum_{m: (t,m)\notin \tau^i} P_i (t,m)\\
&<&\delta/3+\sum_{m: (t,m)\notin \tau^i} P_i (t,m)(1_{\{P_i (t,m)\geq \delta/(3M)\}}+1_{\{P_i (t,m)< \delta/(3M)\}})\\ %1_{\tau^i(t,m)=0}
&<&2\delta/3+\sum_{m: (t,m)\notin \tau^i} 1_{\{P_i (t,m)\geq \delta/(3M)\}}. %1_{\tau^i(t,m)=0}
\end{eqnarray*}
The first inequality follows from the inequality $\sigma(t,m)[C]\leq 1$.
The second inequality follows from the inequality $P_i(1^*_{<t})<\delta/3$ and the equality $1=1_{\{P_i (t,m)\geq \delta/(3M)\}}+1_{\{P_i (t,m)< \delta/(3M)\}}$. Finally, the third inequality follows from the inequalities $\sum_m P_i(t,m)1_{\{P_i(t,m)<\delta/(3M)\}}< \sum_m \delta/(3M)=M\delta/(3M)=\delta/3$ and
$P_i(t,m)1_{\{P_i(t,m)\geq \delta/(3M)\}}\leq 1_{\{P_i(t,m)\geq \delta/(3M)\}}$.

Therefore,
\[r^i_t\geq \delta \implies \sum_{m: (t,m)\notin \tau^i} 1_{\{P_i (t,m)\geq \delta/(3M)\}}>\delta/3.\]
Therefore, for every $t$ such that $r^i_t\geq \delta$, there is $m(t)\in M$ such that $(t,m(t))\notin \tau^i$ and $P_i(t,m(t))\geq \delta/(3M)$.

Set $T_{n,\delta}:=\{(t,m(t)): t\geq n \mbox{ and } r^i_t\geq \delta\}$.

\begin{eqnarray*}\sum_{(t,m(t))\in T_{n,\delta}}\frac{\delta}{3M}\,\s(t,m(t))[A]&\leq& \sum_{(t,m(t))\in T_{n,\delta}}P_i(t,m(t))\s(t,m(t))[A]\\
&\leq & P_i(0^*_{>n})\to_{n\to\infty} 0.
\end{eqnarray*}
Therefore, as $\sum_{(t,m)\in \tau^i}\s(t,m)[A]<\delta/3$, there exists $n_i$ such that \[\sum_{(t,m)\in \tau^i}\s(t,m)[A]+\sum_{(t,m(t))\in T_{n_i,\delta}}\s(t,m(t))[A]<\delta/3.\]

Let $\tau^{i+1}=\tau^i\cup T_{n_i,\delta}$. The pure $(m_t)$-based strategy $\tau^{i+1}$ satisfies condition (\ref{eqn:sumsigmataui}). Note that $T_{n_i,\delta}$ may be the empty set, and in this case $\tau^{i+1}=\tau^i$.

Let $X\subset \mathbb{N}$ with $\infty>|X|>(M+1)/\delta$ and set $t_X=\max_{j\leq i\in X} n_j$. Fix $t>t_X$. For all $i,j\in X$ with $i<j$ and $r^i_t\geq \delta$, there is $m\in M$ with $(m,t)\in \tau^{i+1}\setminus \tau^i\subseteq \tau^j\setminus \tau^i$. As the number of memories is $M$, there are at most $M+1$ elements $i\in X$ with $r^i_t\geq \delta$.
This completes the proof of Lemma \ref{lem:Xstrategies}.
\end{proof}

\begin{proof}[Proof of Theorem \ref{thm:limsupEworthless}]
Fix  a positive integer $M$, a memory process $(m_t)_{t=1}^\infty$ in ${\mathcal M}_M$, $\delta>0$, and an $(m_t)$-based strategy $\s$ of player 1.

Let $\tau^i$, $1\leq i$, be a sequence of pure $(m_t)$-based strategies of player 2, such that for any finite set $X\subset \mathbb{N}$, there is a positive integer $t_\delta$, such that
\begin{equation}\label{eqn:fewbadtaui}
 \forall t\geq t_\delta, \;\;\;|\{i\in X: E_{\s,\tau^i}r_t:=r^i_t>\delta\}|=\sum_{i=1}^\infty1_{\{r^i_t\geq \delta\}}\leq M+1.
\end{equation}
The existence of the such a sequence of strategies $\tau^i$,  $1\leq i<\infty$, that satisfy (\ref{eqn:fewbadtaui}) is guaranteed by Lemma \ref{lem:Xstrategies}.

Fix a finite set $X\subset \mathbb{N}$ with $|X|>(M+1)/\delta$. Let $\tau$ be the uniform mixture of the strategies $\tau^i$, $i\in X$. For every $t$, $r^i_t \leq 1_{\{r^i_t\geq \delta\}}+\delta$, and therefore,
\[\forall t\geq t_\delta, \;\;\;E_{s,\tau}r_t=\frac{1}{|X|}\sum_{i\in X} r^i_t\leq \frac{1}{|X|} \sum_{i\in X}(1_{\{r^i_t\geq \delta\}}+\delta)\leq \frac{M+1}{|X|}+\delta<2\delta.\]
Hence, if $n $ is sufficiently large so that $ t_\delta/n\leq \delta$, then
\[\gamma_n(\s,\tau)=E_{\s,\tau}\frac1n\sum_{t=1}^nr_t\leq \frac{t_\delta}{n}+2\delta<3\delta,\]
which completes the proof of the theorem.
\end{proof}
%%%%%%%%%%%%%%%%%%%%%%%%%%%%%%%%%%%%%%%%%%%%%%%%%%%%%%%%%%%%%%%%%%%%%%%%%%%

%%% Local Variables:
%%% mode: LaTeX
%%% TeX-master: "paper"
%%% End:

\section{Future Directions and Open Problems}
The near-optimal strategies constructed in this paper have three key properties:
\begin{itemize}
\item Public memory (as opposed to private memory).
\item Time-independent action selection (as opposed to time dependent action selection).
\item Time-independent memory updating (as opposed to time dependent memory updating).
\end{itemize}

A clear direction for future research is determining how these properties affect the number of memory states needed for near-optimal strategies. In this directtion, a major open question is whether in any stochastic game there exists a finite-memory strategy that is near-optimal. This problem has been resolved for absorbing games:
\begin{itemize}
\item In any absorbing game, for every $\varepsilon>0$, there exists a private-memory strategy with time-dependent action selection and memory updating that uses only finitely many memory states while being both uniform $\varepsilon$-optimal and $\liminf$ $\varepsilon$-optimal \cite{hanibsney21}.
    \end{itemize}

However, whether such a strategy exists in general stochastic games remains unknown.
Several other important questions remain open about public-memory strategies, both in general stochastic games and in the Big Match:

\begin{itemize}
\item Tightness of $O(\log n)$ Memory: Is the bound of $O(\log n)$ memory states in the first $n$ stages of a public-memory, uniform $\ve$-optimal strategy tight?
\item $\liminf$ and Uniform $\ve$-Optimality: Is there a public-memory strategy that uses at most $O(\log n)$ in the first $n$ stages and is both uniform $\ve$-optimal and $\liminf$ $\ve$-optimal?
    \item Minimal Memory for $\limsup$ Optimality: What is the smallest number of public memory states required in the first $n$ stages for a $\limsup$ $\ve$-optimal strategy?
\end{itemize}

\bibliographystyle{abbrv}
\bibliography{SGwSS}

\end{document}